\documentclass{article}
\usepackage{fullpage}

\usepackage{cite}
\usepackage{algorithmic}
\usepackage{graphicx}
\usepackage{textcomp}
\usepackage{xcolor}
\def\BibTeX{{\rm B\kern-.05em{\sc i\kern-.025em b}\kern-.08em
    T\kern-.1667em\lower.7ex\hbox{E}\kern-.125emX}}

\newcommand\blfootnote[1]{%
	\begingroup
	\renewcommand\thefootnote{}\footnote{#1}%
	\addtocounter{footnote}{-1}%
	\endgroup
}

\usepackage{hyperref}
\hypersetup{pdftex,colorlinks=true,allcolors=blue} 
\usepackage{hypcap} 

\usepackage{amsmath, amsthm, amssymb}
\usepackage{verbatim}
\usepackage{graphicx}
\usepackage{array}
\usepackage{mathrsfs}
\usepackage{comment}
\usepackage{enumitem} 
\usepackage{lmodern}

\newtheorem{theorem}{Theorem}
\newtheorem{lemma}{Lemma}

\newtheorem{definition}{Definition}
\newcommand{\R}{\mathbb{R}}
\newcommand{\E}{\mathbb{E}}

\def\sA{{\mathsf A}}

\def\sS{{\mathsf S}}

\def\sW{{\mathsf W}}
\def\sX{{\mathsf X}}
\def\sY{{\mathsf Y}}

\def\rd{{\rm d}}
\def\PP{{\mathbb P}}

\def\deq{\triangleq}

\begin{document}

\title{
Optimal Feedback Communication with Information Maximization and  Distortion Minimization
}
\date{}

\author{Aolin Xu}
\maketitle
\blfootnote{xuaolin@gmail.com}

\begin{abstract}
We study the problem of optimally sending a real-valued source through multiple uses of a channel with feedback. First, we state a set of conditions that are sufficient for an encoder to achieve maximal mutual information between the source and all the channel outputs. This set of conditions are also necessary when the channel is input-identifiable, a condition widely satisfied by common channel models.
More notably, we further study the information maximization-distortion minimization problem, where the mutual information between the source and all channel outputs still needs to be maximized, while at each step, the MMSE of estimating the source from the channel outputs so far also needs to be minimized. We derive a solution to this problem for discrete channels with certain symmetries, e.g. $k$-ary symmetric or $k$-ary erasure channels. We show that for such channels the famous posterior matching scheme, while not necessary for information maximization alone, is sufficient and essentially necessary for achieving both information maximization and distortion minimization.
This work also provides a new perspective of regularizing distortion-minimizing feedback communication through information maximization, which enables us to find the optimal solution that otherwise would be intractable.
\end{abstract}

\section{Introduction}
In a landmark work by Shayevitz and Feder \cite{pm2011}, the posterior matching scheme was proposed as 
a fundamental principle for feedback communication and as 
a simple sequential scheme 
for 
achieving reliable communication at any rate below the channel capacity.
The well-known Schalkwijk-Kailath scheme for the
AWGN channel \cite{SchalkwijkKailath,Schalkwijk66} and the Horstein
scheme for the BSC \cite{Horstein63}  can be derived as special cases of posterior matching.
If the goal for communication is merely achieving maximal mutual information between the source and the channel outputs, the posterior matching scheme is sufficient, but is not necessary. 
This work is partly motivated by what posterior matching can achieve beyond maximal mutual information and reliable communication at capacity.

We propose the problem of optimal feedback communication with information maximization-distortion minimization, where the mutual information between the source and all channel outputs still needs to be maximized, while at each step, the MMSE of estimating the source from the channel outputs so far also needs to be minimized. We derive a solution to this problem for discrete channels with certain symmetries, e.g. $k$-ary symmetric or $k$-ary erasure channels. We show that for such channels the posterior matching scheme is sufficient and essentially necessary for achieving both information maximization and distortion minimization.

From an operational perspective, information maximization may be viewed as a form of regularization for distortion minimization. Without this regularization, it is generally intractable to solve the distortion minimization problem by itself, even when the channel is noiseless. With this regularization, however, our results show that a clean solution to distortion minimization can be obtained under certain symmetry conditions of the channel.


In the literature, there are studies of joint source-channel coding with feedback other than \cite{SchalkwijkKailath, Schalkwijk66, Horstein63, pm2011}, including works on fixed blocklength coding \cite{Gastpar_jscc_fb03}, variable length coding \cite{iter_fb02}, minimizing accumulated distortion \cite{dynamic_jscc_fb13}, fundamental limits of delay-distortion trade-off \cite{vic_pol_ver_jscc17}, and control-theoretic viewpoint \cite{control_comm_fb12}, but none of them studies the problem considered here.
This work potentially bridges the information-theoretic feedback communication and control-oriented communication. Our findings justify the use of posterior matching in real-time applications where immediate state estimation is critical, such as networked control systems \cite{anytime_capacity06} and human-computer interfaces \cite{coleman_hci}.
Moreover, it provides a new perspective of regularizing distortion-minimizing feedback communication through information maximization, which enables us to find the optimal solution.






\section{Problem statements}
Let $W$ be a continuous random variable with a distribution $P_W$ over $\sW\subset\R$ that can be described by a continuous cumulative distribution function (CDF) $F_W$ or a probability density function (PDF) $p_W$.
The goal is to design an encoder that can sequentially send $W$ through a channel with input alphabet $\sX$, output alphabet $\sY$ and probability transition law $P_{Y|X}$, where the channel output after each channel use is fed back to the sender, such that after $n$ channel uses, the mutual information between $W$ and all channel outputs is maximized, while after each channel use, the MMSE of estimating $W$ from the channel outputs so far is minimized. 

\begin{figure}[t]
    \centering
\includegraphics[width=0.25\textwidth]{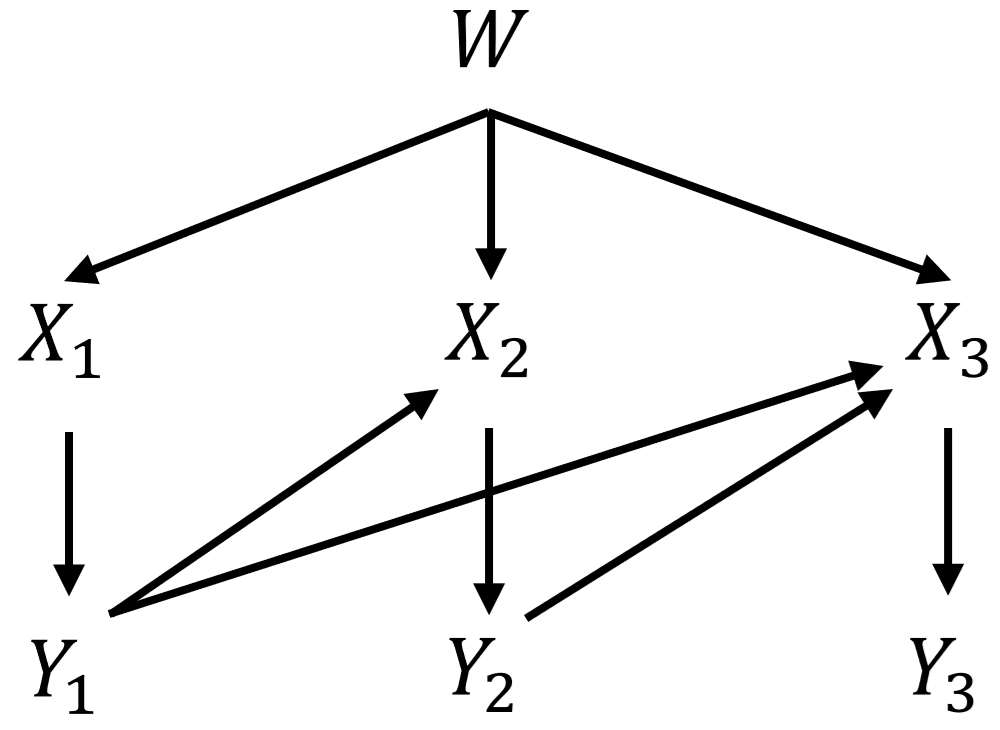}
    \caption{Causal relationships of the random variables in feedback communication.}
    \label{fig:feedback}
\end{figure}

Formally, the encoder can be described by a sequence of encoding functions $(\varphi_t)_{t=1}^n$, with $\varphi_t: \sW\times\sY^{t-1}\rightarrow\sX$.
At the $t$th step, the encoder maps the source $W$ together with the previous channel outputs $Y^{t-1}\deq(Y_1, \ldots, Y_{t-1})$ to a new channel input as $X_t=\varphi_t(W,Y^{t-1})$. After sending $X_t$ through the channel, the new channel output $Y_t$, generated based on $X_t$ through the probability transition law $P_{Y|X}$, is fed back to the encoder. 
The Bayesian network of the random variables in this feedback communication setup is drawn according to their causal relationships and is shown in Fig.~\ref{fig:feedback}.
We aim to solve 
two problems.
\begin{itemize}[leftmargin=*]
    \item Information maximization
    
Given $P_W$, $P_{Y|X}$ and $n$, the information maximization problem can be stated as
\begin{align}
    \max_{\varphi_t:\,\sW\times\sY^{t-1} \rightarrow \sX , \, t=1,\ldots,n} I(W; Y^n) . \label{eq:opt_I}
\end{align}
Denote the set of the sequences of encoding functions achieving the maximum of \eqref{eq:opt_I} as $\mathsf\Phi$. 

    \item Information maximization-distortion minimization

At the $t$th step, given $Y^{t-1}$, any encoding function $\varphi_t:\,\sW\times\sY^{t-1} \rightarrow \sX$ induces a $Y^{t-1}$-dependent encoding function 
\begin{align}
\phi_{t,Y^{t-1}}:\,\sW \rightarrow \sX, \, W\mapsto \varphi_t(W, Y^{t-1}) . \label{eq:Y-depend_phi}
\end{align}
For each $t=1,\ldots,n$, given $P_W$, $P_{Y|X}$, $Y^{t-1}$ and $\mathsf\Phi$, the information maximization-distortion minimization problem can be stated as
\begin{align}
    \min_{\substack{\phi_{t,Y^{t-1}}:\,\sW \rightarrow \sX , \, W\mapsto \varphi_t(W, Y^{t-1})\\ \text{s.t. } (\varphi_1, \ldots, \varphi_n)\in\mathsf\Phi}} \text{mmse}(W|Y^{t}) \label{eq:opt_mmse}
\end{align}
where 
\begin{align}
\text{mmse}(W|Y^{t}) \deq \min_{\psi_t:\,\sY^t\rightarrow\sW}\E[(W-\psi_t(Y^{t}))^2]    \label{eq:mmse_W|Yt_def}
\end{align}
for $t = 1, \ldots, n$.
The function $\psi_t$ in \eqref{eq:mmse_W|Yt_def} can be viewed as the optimal decoder at the $t$th step that minimizes the mean squared error of estimating $W$ from $Y^t$. It does not affect any random variables under consideration. 
\end{itemize}
For the information maximization problem in \eqref{eq:opt_I}, we seek a sequence of encoding functions $(\varphi_t)_{t=1}^n$ that can maximize the mutual information between $W$ and the overall channel outputs $Y^n$ after $n$ steps; while for the information maximization-distortion minimization problem, we maintain the maximal mutual information $I(W;Y^n)$, and additionally at each step, given $Y^{t-1}$, we seek a $Y^{t-1}$-dependent encoder $\phi_{t, Y^{t-1}}$ to greedily minimize ${\rm mmse}(W|Y^t)$ for that step.

\section{Optimal encoder for information maximization}
Throughout, we assume the channel $P_{Y|X}$ has information capacity $C$ achieved by an input distribution $P_X^*$. Following the notation in \cite{polyanskiy_Wu_book}, this means
\begin{align}
    \max_{P_X } I(P_X, P_{Y|X}) = I(P_X^*, P_{Y|X}) = C . \label{eq:info_C_def}
\end{align}
For channels with finite information capacity only when the input distribution is constrained, e.g. AWGN channel with constrained second order moment of the input, we assume such a constraint exists and is applied to the maximization of $P_X$ in \eqref{eq:info_C_def}.
In the following two subsections, we first state a set of conditions that are sufficient for the encoding functions to achieve the maximal $I(W;Y^n)$ as $nC$, and then state that under certain assumptions on the channel, these conditions are also necessary. The proofs of the results in this section can be found in \cite{info_max_feedback}.
\subsection{Sufficiency}
\begin{theorem}\label{th:suff_info_max}
If the sequence of encoding functions $(\varphi_t)_{t=1}^n$ satisfy the following three conditions for all $t=1,\ldots,n$: 
\begin{enumerate}[leftmargin=*]
    \item 
    $\varphi_t$ is a deterministic function, meaning that $X_t$ is uniquely determined by $(W, Y^{t-1})$;
    \item 
    the marginal distribution of $X_t$ is a capacity-achieving distribution $P_X^*$ of the channel;
    \item 
    $X_t$ is statistically independent of $Y^{t-1}$;
\end{enumerate}
then $I(W;Y^n)$ achieves the maximal value $nC$.
\end{theorem}

\subsection{Necessity}
First, we introduce two conditions of a channel, which are widely satisfied by commonly used channel models.
\begin{definition}
A channel $P_{Y|X}$ is injective if the distributions $(P_{Y|X=x}, x\in\sX)$ are all different.
\end{definition}
\begin{definition}
A channel $P_{Y|X}$ is input-identifiable if the distributions $(P_{Y|X=x}, x\in\sX)$ are linearly independent.
\end{definition}
All practically meaningful channel models are injective. Being injective is also a necessary condition for a channel to be input-identifiable.
When a channel is input-identifiable, different input distributions map to different output distributions, and an input distribution $P_X$ can be identified from the output distribution $P_Y$.
Commonly used discrete channel models, including binary symmetric channel, $k$-ary symmetric channel, binary erasure channel and $k$-ary erasure channel are all input-identifiable. The concept of input-identifiable channels can be extended to include AWGN channel, additive exponential noise channel and Poisson channel.
The following two lemmas are stated under the two channel conditions respectively.
\begin{lemma}\label{lm:ch_non_degrade}
For random variables $X$, $Y$ and $Z$ forming a Markov chain $Z-X-Y$, under the assumption that $P_{Y|X}$ is injective, if $I(Z;Y) = I(X;Y)$, then $X$ needs to be a function of $Z$.
\end{lemma}

\begin{lemma}\label{lm:ch_identif}
For random variables $X$, $Y$ and $Z$ forming a Markov chain $Z-X-Y$, under the assumption that $P_{Y|X}$ is input-identifiable, if $Y$ is statistically independent of $Z$, then $X$ needs to be statistically independent of $Z$.
\end{lemma}

\begin{theorem}\label{th:nece_info_max}
Under the assumption that the channel $P_{Y|X}$ is input-identifiable, the three conditions listed in Theorem~\ref{th:suff_info_max} are also necessary for the sequence of encoding functions $(\varphi_t)_{t=1}^n$ to satisfy for all $t=1,\ldots,n$ in order to have $I(W;Y^n) = nC$.
\end{theorem}

\section{Optimal encoder for information maximization-distortion minimization}


Using basic properties of MMSE and conditional expectation, it can be shown that
\begin{align}
\text{   mmse}(W|Y^t) =
    \E[W^2]- \E[\E[W|Y^t]^2] .
\end{align}
As $\E[W^2]$ is determined by $P_W$, whenever $\E[W^2]<\infty$, the information maximation-distortion minimization problem in \eqref{eq:opt_mmse} is equivalent to
\begin{align}
    \max_{\substack{\phi_{t,Y^{t-1}}:\,\sW \rightarrow \sX, \, W\mapsto \varphi_t(W, Y^{t-1}) \\ \text{s.t. } (\varphi_1, \ldots, \varphi_n)\in\mathsf\Phi}} \E[\E[W|Y^{t}]^2] \label{eq:opt_EW|Y]^2}
\end{align}
To derive the solution to the optimization problem in \eqref{eq:opt_EW|Y]^2}, we first prove two useful properties about $\E[W|Y^t]$ when the encoder achieves information maximization.
Throughout this section, we assume $P_{Y|X}$ is a discrete memoryless channel (DMC) with input alphabet $\sX=\{1,\ldots,k\}$ and output alphabet $\sY=\{1,\ldots,l\}$. Also denote the capacity-achieving output distribution of the channel as $P_Y^* = P_{Y|X}\circ P_X^*$.
\begin{lemma}\label{lm:E[W|Y]}
Under the assumption that the channel $P_{Y|X}$ is input-identifiable, with a sequence of encoding functions belonging to $\mathsf\Phi$, we have for each $t=1,\ldots,n$,
\begin{align}
& \E[\E[W|Y^{t}]^2] 
=\sum_{y^{t-1}}P_{Y^{t-1}}(y^{t-1}) \E[\E[W|Y^{t-1}=y^{t-1}, Y_t]^2] \label{eq:lm_E[W|Yt]_1}
\end{align}
and
\begin{align}
& \E[W|Y^{t-1}=y^{t-1}, Y_t = y_t] = 
\sum_{x_t}\frac{P_{X}^*(x_t) P_{Y|X}(y_t|x_t)}{P_{Y}^*(y_t)} \E[W|Y^{t-1}=y^{t-1}, X_t=x_t] . \label{eq:lm_E[W|Yt]_2}
\end{align}
\end{lemma}
\begin{proof}
Under the assumptions that the channel $P_{Y|X}$ is input-identifiable and the encoder belongs to $\mathsf\Phi$, Theorem~\ref{th:nece_info_max} implies that $X_t$ and hence $Y_t$ are statistically independent of $Y^{t-1}$.
Equation~\eqref{eq:lm_E[W|Yt]_1} follows from
\begin{align}
\E[\E[W|Y^{t}]^2]  
&= \sum_{y^{t-1}}P_{Y^{t-1}}(y^{t-1}) \E[\E[W|Y^{t-1}=y^{t-1}, Y_t]^2|Y^{t-1}=y^{t-1}] \nonumber\\
&= \sum_{y^{t-1}}P_{Y^{t-1}}(y^{t-1}) \E[\E[W|Y^{t-1}=y^{t-1}, Y_t]^2] \label{eq:lm_E[W|Yt]_3}
\end{align}
where \eqref{eq:lm_E[W|Yt]_3} follows from the necessary condition that $Y_t$ is statistically independent of $Y^{t-1}$, as a consequence of the above assumptions.

Equation~\eqref{eq:lm_E[W|Yt]_2} follows from
\begin{align}
\E[W|Y^{t-1}=y^{t-1}, Y_t=y_t]  
&= \sum_{x_t} P_{X_t|Y^{t-1}, Y_t}(x_t|y^{t-1}, y_t)
\E[W|Y^{t-1}=y^{t-1}, Y_t=y_t, X_t=x_t] \\
&= \sum_{x_t}\frac{P_{X}^*(x_t) P_{Y|X}(y_t|x_t)}{P_{Y}^*(y_t)} \E[W|Y^{t-1}=y^{t-1}, X_t=x_t] 
\end{align}
where the last step follows from
\begin{align}
P_{X_t|Y^{t-1}, Y_t}(x_t|y^{t-1}, y_t) 
&= \frac{P_{X_t|Y^{t-1}}(x_t|y^{t-1})P_{Y_t|X_t}(y_t|x_t)}{P_{Y_t|Y^{t-1}}(y_t|y^{t-1})}  \label{eq:lm_E[W|Yt]_5}\\
&= \frac{P_{X_t}(x_t) P_{Y_t|X_t}(y_t|x_t) }{P_{Y_t}(y_t)}  \label{eq:lm_E[W|Yt]_6} \\
&= \frac{P_{X}^*(x_t) P_{Y|X}(y_t|x_t)}{P_{Y}^*(y_t)}  \label{eq:lm_E[W|Yt]_7}
\end{align}
and from
\begin{align}
&\E[W|Y^{t-1}=y^{t-1}, Y_t=y_t, X_t=x_t] = 
\E[W|Y^{t-1}=y^{t-1}, X_t=x_t]  \label{eq:lm_E[W|Yt]_8}
\end{align}
where \eqref{eq:lm_E[W|Yt]_5} follows from the Markov chain $Y^{t-1}-X_t-Y_t$; 
\eqref{eq:lm_E[W|Yt]_6} follows from the necessary condition that $X_t$ and $Y_t$ are statistically independent of $Y^{t-1}$ under the assumptions;
\eqref{eq:lm_E[W|Yt]_7} follows from the fact that $P_{Y_t|X_t}=P_{Y|X}$ and the necessary condition $P_{X_t}=P_X^*$ hence $P_{Y_t}=P_Y^*$ under the assumptions; 
and 
\eqref{eq:lm_E[W|Yt]_8} follows from the Markov chain $W - (Y^{t-1},X_t) - Y_t$.
\end{proof}
Equation~\eqref{eq:lm_E[W|Yt]_1} in Lemma~\ref{lm:E[W|Y]} implies that under the input-identifiability assumption of the channel, the optimization problem in \eqref{eq:opt_EW|Y]^2} can be solved by maximizing $\E[\E[W|Y^{t-1}=y^{t-1}, Y_t]^2]$ for each given $Y^{t-1}=y^{t-1}$. In view of this, we derive a compact yet revealing form of $\E[\E[W|Y^{t-1}=y^{t-1}, Y_t]^2]$.

\begin{lemma}\label{lm:E[W|Y]^2}
Under the assumption that the channel $P_{Y|X}$ is input-identifiable, with a sequence of encoding functions belonging to $\mathsf\Phi$, we have for each $t=1,\ldots,n$,
\begin{align}
 \E[\E[W|Y^{t-1}=y^{t-1}, Y_t]^2] = {\boldsymbol b}_{y^{t-1}}^\top {\boldsymbol K} {\boldsymbol b}_{y^{t-1}} \label{eq:bKb_form}
\end{align}
where ${\boldsymbol b}_{y^{t-1}}$ is a length-$k$ column vector with the $i$th element
\begin{align}
{\boldsymbol b}_{y^{t-1}}(i) \deq P_{X}^*(i) \E[W|Y^{t-1}=y^{t-1}, X_t=i] \label{eq:b_def}
\end{align}
and ${\boldsymbol K}$ is a $k\times k$ matrix with the $(i,j)$th element
\begin{align}
{\boldsymbol K}(i,j) \deq \sum_{y_t=1}^l \frac{P_{Y|X}(y_t|i) P_{Y|X}(y_t|j)}{P_{Y}^*(y_t)}  . \label{eq:K_def}
\end{align}
\end{lemma}
\begin{proof}
We have
\begin{align}
\E[\E[W|Y^{t-1}=y^{t-1}, Y_t]^2] 
=& \sum_{y_t=1}^l P_{Y}^*(y_t) \E[W|Y^{t-1}=y^{t-1}, Y_t = y_t]^2 \\
=  
& \sum_{i=1}^k P_{X}^*(i) \E[W|Y^{t-1}=y^{t-1}, X_t=i] \cdot \nonumber \\
& \quad \sum_{j=1}^k P_{X}^*(j) \E[W|Y^{t-1}=y^{t-1}, X_t=j]  \cdot \nonumber \\
& \qquad \sum_{y_t=1}^l \frac{ P_{Y|X}(y_t|i) P_{Y|X}(y_t|j)}{P_{Y}^*(y_t)} \label{eq:lm:bKb_1} \\
= &
{\boldsymbol b}_{y^{t-1}}^\top {\boldsymbol K} {\boldsymbol b}_{y^{t-1}} 
\end{align}
where \eqref{eq:lm:bKb_1} follows from \eqref{eq:lm_E[W|Yt]_2} in Lemma~\ref{lm:E[W|Y]} and rearranging the order of summations.
\end{proof}
The matrix $\boldsymbol{K}$ in the quadratic form \eqref{eq:bKb_form} solely depends on the channel $P_{Y|X}$. Under certain conditions of $P_{Y|X}$, this quadratic form is Schur convex. Lemma~\ref{lm:E[W|Y]^2} thus reveals a route to use majorization to maximize $\E[\E[W|Y^{t-1}=y^{t-1}, Y_t]^2]$, as stated in the following lemma.
\begin{lemma}\label{lm:schur_cvx}
Under the assumption that the channel $P_{Y|X}$ results in a matrix $\boldsymbol{K}$ in \eqref{eq:K_def} of the form $\boldsymbol{K}=\alpha \boldsymbol{I} + \beta \boldsymbol{1}\boldsymbol{1}^\top$, if a $k$-length vector ${\boldsymbol b}$ majorizes ${\boldsymbol b}'$, then ${\boldsymbol b}^\top \boldsymbol{K} {\boldsymbol b} \ge {\boldsymbol b}'^\top \boldsymbol{K} {\boldsymbol b}'$.
\end{lemma}
\begin{proof}
When $\boldsymbol{K}=\alpha \boldsymbol{I} + \beta \boldsymbol{1}\boldsymbol{1}^\top$ with $\alpha\ge 0$, it has equal diagonal elements and equal off-diagonal elements, resulting ${\boldsymbol b}^\top \boldsymbol{K} {\boldsymbol b}$ as a convex function of $\boldsymbol{b}$ and invariant to the order of the elements in $\boldsymbol{b}$, hence it is Schur convex in $\boldsymbol{b}$ \cite[Example 3.8]{SchurConvex07}. 
The condition $\alpha\ge 0$ can always be satisfied as $\boldsymbol{K}$ is entrywise nonnegative.
The claim follows from the property of Schur convex functions.
\end{proof}

Next, we optimize over $\boldsymbol{b}_{y^{t-1}}$ through majorization. The route is similar to the optimality proof technique developed in \cite{quantization_xu}. 
Under the assumption that $W$ is a continuous random variable taking values in $\sW\subset\R$, whenever $\PP[Y^{t-1}=y^{t-1}]>0$, the conditional CDF $F_{W|Y^{t-1}=y^{t-1}}$ is continuous and increasing on its support hence it has a continuous and increasing inverse function
\begin{align}
Q_{W|y^{t-1}}(u) \deq F_{W|Y^{t-1}=y^{t-1}}^{-1}(u), \quad u\in[0,1] .
\end{align}
With a random variable $U$ uniformly distributed over $[0,1]$, the random variable $Q_{W|y^{t-1}}(U)$ has the distribution $P_{W|Y^{t-1}=y^{t-1}}$ for all $y^{t-1}$. 
At the $t$th step, given $Y^{t-1}=y^{t-1}$, a $Y^{t-1}$-dependent encoding function $\phi_{t,Y^{t-1}}:\,\sW \rightarrow \sX$ defined in \eqref{eq:Y-depend_phi} induces a $y^{t-1}$-dependent labeled partition of $[0,1]$, denoted as $(\sS_{y^{t-1}}(1), \ldots, \sS_{y^{t-1}}(k))$ with
\begin{align}
\sS_{y^{t-1}}(i) \deq \{u\in[0,1]: \phi_{t,y^{t-1}}(Q_{W|y^{t-1}}(u))=i\} .
\end{align}
We have
\begin{align}
& \E[W|Y^{t-1} = y^{t-1}, X_t = i] \nonumber \\
= & \E[Q_{W|y^{t-1}}(U)|\phi_{t,y^{t-1}}(Q_{W|y^{t-1}}(U)) = i] \\
=& \frac{1}{P_X^*(i)} \int_{\sS_{y^{t-1}}(i)} Q_{W|y^{t-1}}(u) \rd u 
\end{align}
hence an equivalently $(\sS_{y^{t-1}}(1), \ldots, \sS_{y^{t-1}}(k))$-induced 
\begin{align}
\boldsymbol{b}_{y^{t-1}}(i) = 
\int_{\sS_{y^{t-1}}(i)} Q_{W|y^{t-1}}(u) \rd u \quad i=1,\ldots,k .
\end{align}
The following lemma shows that when $P_X^*$ is uniform, we can find the majorizing $\boldsymbol{b}_{y^{t-1}}$ by choosing $(\sS_{y^{t-1}}(1), \ldots, \sS_{y^{t-1}}(k))$ as contiguous intervals.
\begin{lemma}\label{lm:major_b}
When $P_X^*$ is uniform over $\{1,\ldots,k\}$,
among the category of labeled partitions $(\sS_{y^{t-1}}(1),\ldots,\sS_{y^{t-1}}(k))$ of $[0,1]$ resulting in 1) $P_U(\sS_{y^{t-1}}(i))=P_X^*(i)$, and 2) an order of $\boldsymbol{b}_{y^{t-1}}(\sigma(1)) \le \ldots \le \boldsymbol{b}_{y^{t-1}}(\sigma(k))$ specified by a permutation $\sigma$ of $\{1,\ldots,k\}$, the one defined by equal-length contiguous intervals
\begin{align}
    \sS^*(\sigma(i)) \deq [q_{i-1}, q_i)  \quad\text{with $q_i \deq q_{i-1} + P_X^*(\sigma(i))$} \label{eq:def_S*_unif_PX}
\end{align}
for $i=1,\ldots,k$ and $q_0 \deq 0$ induces a vector $\boldsymbol{b}^*_{y^{t-1}}$ satisfying 
\begin{align}
\sum_{i=1}^j \boldsymbol{b}_{y^{t-1}} ^*(\sigma(i))
    \le
    \sum_{i=1}^j \boldsymbol{b}_{y^{t-1}}(\sigma(i)) \quad j=1,\ldots, k-1 \label{eq:mj_b_1}
\end{align}
and 
\begin{align}
\sum_{i=1}^k \boldsymbol{b}_{y^{t-1}}^*(\sigma(i))
    =
    \sum_{i=1}^k \boldsymbol{b}_{y^{t-1}}(\sigma(i)) \label{eq:mj_b_2}
\end{align}
for $\boldsymbol{b}_{y^{t-1}}$ induced by any other labeled partitions of $[0,1]$ in this category.
\end{lemma}
\begin{proof}
First, as $P_X^*$ is uniform and $Q_{W|y^{t-1}}(u)$ is increasing in $u$, $\boldsymbol{b}_{y^{t-1}}^*$ induced by the construction of $(\sS^*(1),\ldots, \sS^*(k))$ in \eqref{eq:def_S*_unif_PX} satisfies the order specified by $\sigma$.
For any other $(\sS_1,\ldots,\sS_k)$ in this category of labeled partitions, let $\sA_j \deq \cup_{i=1}^j \sS^*({\sigma(i)})$ and ${\mathsf B}_j \deq \cup_{i=1}^j \sS(\sigma(i))$. Then $P_U(\sA_j) = P_U(\mathsf B_j)$, hence
\begin{align}
    P_U(\sA_j\setminus\mathsf B_j) = P_U(\mathsf B_j \setminus \sA_j) . \label{eq:mj_b_B-A=A-B}
\end{align}
It follows that
\begin{align}
& \sum_{i=1}^j \boldsymbol{b}_{y^{t-1}} (\sigma(i)) 
-
\sum_{i=1}^j \boldsymbol{b}_{y^{t-1}} ^*(\sigma(i)) \nonumber \\
=& 
\int_{\mathsf B_j} Q_{W|y^{t-1}}(u) \rd u 
-
\int_{\sA_j} Q_{W|y^{t-1}}(u) \rd u \\
=& \int_{\mathsf B_j\setminus\sA_j} Q_{W|y^{t-1}}(u) \rd u 
-
\int_{\sA_j\setminus\mathsf B_j} Q_{W|y^{t-1}}(u) \rd u \\
\ge& P_U(\mathsf B_j \setminus \sA_j) \Big( \inf_{u\in \mathsf B_j \setminus \sA_j} Q_{W|y^{t-1}}(u) - \sup_{u\in \mathsf A_j \setminus \mathsf B_j} Q_{W|y^{t-1}}(u)\Big) \label{eq:major_b_B-A} \\
\ge & 0  \label{eq:major_b_B>A}
\end{align}
where \eqref{eq:major_b_B-A} follows from \eqref{eq:mj_b_B-A=A-B} and the fact that $Q_{W|y^{t-1}}(u)$ is increasing in $u$, and \eqref{eq:major_b_B>A} follows from the fact that $Q_{W|y^{t-1}}(u)$ is increasing in $u$ and 
\begin{align}
    \sup( \sA_j\setminus\mathsf B_j )
\le 
    \inf( \mathsf B_j\setminus\sA_j )
\end{align}
by the construction of $\sS^*({\sigma(i))}$ for $i=1,\ldots,j$ in \eqref{eq:def_S*_unif_PX}.
This proves \eqref{eq:mj_b_1}.
The proof of \eqref{eq:mj_b_2} is straightforward as
$    \sum_{i=1}^k \boldsymbol{b}_{y^{t-1}} ^*
    =
    \sum_{i=1}^k \boldsymbol{b}_{y^{t-1}} 
    = \int_{[0,1]} Q_{W|y^{t-1}}(u) \rd u  = \E[W|Y^{t-1}=y^{t-1}]$.
\end{proof}
It is worthwhile to notice that the labeled partition constructed in \eqref{eq:def_S*_unif_PX} as contiguous intervals does not depend on $y^{t-1}$; moreover, for different $\sigma$, the contiguous intervals are the same and only their labels differ.
With Lemma~\ref{lm:E[W|Y]}, \ref{lm:E[W|Y]^2}, \ref{lm:schur_cvx} and \ref{lm:major_b}, we arrive at the solution to the information maximization-distortion minimization problem.
\begin{theorem}\label{th:info_max_dist_min}
When a discrete memoryless channel $P_{Y|X}$ satisfies the following three conditions:
\begin{enumerate}
    \item the channel is input-identifiable;
    \item the capacity-achieving input distribution $P_X^*$ is uniform over $\{1,\ldots,k\}$;
    \item the matrix $\boldsymbol{K}$ defined in \eqref{eq:K_def} induced by the channel has equal diagonal elements and equal off-diagonal elements;
\end{enumerate}
at the $t$th step, for $t=1,\ldots,n$, a $Y^{t-1}$-dependent encoding function taking the form
\begin{align}
    X_t = \sigma(F_{\sigma^{-1}(X)}^{*-1}(F_{W|Y^{t-1}}(W|Y^{t-1}))) , \label{eq:enc_info_max_d_min}
\end{align}
where $\sigma$ is an arbitrary permutation of $\{1,\ldots,k\}$ and $F_{\sigma^{-1}(X)}^{*-1}$ is the generalized inverse CDF of $\sigma^{-1}(X)$ with $X\sim P_X^*$, achieves the minimum in \eqref{eq:opt_mmse}, that is, it solves the information maximization-distortion minimization problem. 
Moreover, the form of \eqref{eq:enc_info_max_d_min} includes all encoding functions that achieve the minimum in \eqref{eq:opt_mmse} when the source distribution is supported on an interval and the channel satisfies the above conditions.
\end{theorem}
\begin{proof}
From the assumptions of the channel and  Lemma~\ref{lm:E[W|Y]}, \ref{lm:E[W|Y]^2}, \ref{lm:schur_cvx} and \ref{lm:major_b}, we know that at the $t$th step, for a particular order specified by $\sigma$ of the elements in $\boldsymbol{b}_{Y^{t-1}}$ as defined in \eqref{eq:b_def}, $\E[\E[W|Y^{t}]^2]$ can be maximized by a $Y^{t-1}$-dependent encoding function taking the form
\begin{align}
    \phi_{t,Y^{t-1}}(w) = \sigma(i) \quad\text{if } w\in Q_{W|Y^{t-1}}(\sS^*(\sigma(i))) . \label{eq:phi_info_max_d_min}
\end{align}
As 
the labeled partition $(\sS^*(1),\ldots,\sS^*(k))$ constructed in \eqref{eq:def_S*_unif_PX} for different $\sigma$ result in the same contiguous intervals, the above $\phi_{t,Y^{t-1}}(w)$ achieves the global maximum of $\E[\E[W|Y^{t}]^2]$.
To see why this encoding function takes the particular form of \eqref{eq:enc_info_max_d_min}, note that
\begin{align}
w\in Q_{W|Y^{t-1}}(\sS^*(\sigma(i)))
\Leftrightarrow F_{W|Y^{t-1}}(w|Y^{t-1}) \in 
\sS^*(\sigma(i)) , \nonumber
\end{align}
$F_{W|Y^{t-1}}(W|Y^{t-1})$ has the same distribution as $U$, and the inverse CDF of the random variable $\sigma^{-1}(X)$ with $X\sim P_X^*$ maps the contiguous intervals
$(\sS^*({\sigma(1)}), \ldots, \sS^*(\sigma(k)))$ of $[0,1]$ to the indices $(1,\ldots,k)$.
With these observations, we know the encoding function in \eqref{eq:phi_info_max_d_min} is equivalent to \eqref{eq:enc_info_max_d_min}.
The claim that this form of encoding functions uniquely maximizes $\E[\E[W|Y^{t}]^2]$ is due to the Schur convexity of $ {\boldsymbol b}_{y^{t-1}}^\top {\boldsymbol K} {\boldsymbol b}_{y^{t-1}} $ in ${\boldsymbol b}_{y^{t-1}}$ and the form of encoding function encompasses all orders of ${\boldsymbol b}_{y^{t-1}}$.

It remains to show that the sequence of encoding functions of \eqref{eq:enc_info_max_d_min} belongs to $\mathsf\Phi$, i.e.\ it solves the information maximization problem. 
This can be seen that this encoder satisfies the three conditions in Theorem~\ref{th:suff_info_max} sufficient for information maximization: $X_t$ is determined by $(W,Y^{t-1})$, $X_t\sim P_X^*$, and $F_{W|Y^{t-1}}(W|Y^{t-1})$ is independent of $Y^{t-1}$ hence is $X_t$.
This completes the proof.
\end{proof}

It can be verified that $k$-ary symmetric channels and $k$-ary erasure channels satisfy all three conditions in Theorem~\ref{th:info_max_dist_min}.
Finally, with Theorem~\ref{th:info_max_dist_min}, we can see that under these channels, the posterior matching scheme
\begin{align}
    X_t = F_X^{*-1}(F_{W|Y^{t-1}}(W|Y^{t-1})) \label{eq:posterior_matching}
\end{align}
is sufficient to achieve information maximization and distortion minimization; it is essentially necessary as well, as all other encoders that solve this problem are merely relabeling the output of the posterior matching scheme.

{
    \small
    \bibliography{main}
}

\end{document}